\documentclass[twoside,leqno,twocolumn]{article}
\pdfoutput=1
\usepackage{ltexpprt}
\usepackage{graphicx}
\usepackage{amsmath}
\usepackage{amssymb}
\usepackage{booktabs}
\usepackage{multirow}
\usepackage{graphicx}
\usepackage{subcaption}
\usepackage{balance} % to balance last-page columns
\usepackage[hide]{chato-notes}

\newcommand{\parentstransitions}{{\sc ParentTransitions}}
\newcommand{\edgetransitions}{{\sc EdgeTransitions}}
\newcommand{\childrentransitions}{{\sc ChildrenTransitions}}
\newcommand{\nodeitems}{{\sc NodeItems}}

\newcommand{\shortparentstransitions}{{\sc PT}}
\newcommand{\shortedgetransitions}{{\sc ET}}
\newcommand{\shortchildrentransitions}{{\sc CT}}
\newcommand{\shortnodeitems}{{\sc NI}}

\newcommand{\answer}{\ensuremath{\tt A}}

\newcommand{\mcproblem}{{\textsc{Markov Chain Monitoring}}}
\newcommand{\nodeproblem}{{\textsc{Node-Monitoring}}}
\newcommand{\edgeproblem}{{\textsc{Edge-Monitoring}}}
\newcommand{\variant}[1]{{#1}}

\newcommand{\nodegreedy}{{\texttt{NodeGreedy}}}
\newcommand{\edgegreedy}{{\texttt{EdgeGreedy}}}
\newcommand{\edgeDP}{{\texttt{EdgeDP}}}

\newcommand{\hubway}{{\texttt{Hubway}}}
\newcommand{\autonomoussystems}{{\texttt{AS}}}
\newcommand{\ba}{{\texttt{BA}}}
\newcommand{\grid}{{\texttt{Grid}}}
\newcommand{\geo}{{\texttt{Geo}}}

\newcommand{\uniform}{{\tt Uniform}}
\newcommand{\direct}{{\tt Direct}}
\newcommand{\inverse}{{\tt Inverse}}
\newcommand{\ego}{{\tt Ego}}

\newcommand{\indegree}{{\tt In-Degree}}
\newcommand{\nodenumitems}{{\tt Node-NumItems}}
\newcommand{\edgenumitems}{{\tt Edge-NumItems}}
\newcommand{\inprobability}{{\tt In-Probability}}
\newcommand{\nodebetweenness}{{\tt Node-Betweenness}}
\newcommand{\edgebetweenness}{{\tt Edge-Betweenness}}
\newcommand{\closeness}{{\tt Closeness}}
\newcommand{\probability}{{\tt Probability}}

\newcommand{\markovchain}{{\textit{Markov chain}}}
\newcommand{\expecteduncertainty}{{\textit{expected uncertainty}}}

\newcommand{\children}{\ensuremath{\kappa}}
\newcommand{\transition}{\ensuremath{\mathbf{P}}}

\newcommand{\initial}{\ensuremath{\mathbf{x}}}
\newcommand{\final}{\ensuremath{\mathbf{z}}}
\newcommand{\uncertainty}{\ensuremath{F}}
\newcommand{\objective}{\ensuremath{\mathbf{F}}}
\newcommand{\variance}{\ensuremath{\mathbf{var}}}
\newcommand{\rvFinal}{\ensuremath{\mathbf{Z}}}

\newcommand{\prob}{\ensuremath{\mathbf{Pr}}}

\DeclareMathOperator*{\argmin}{arg\,min}

\newcommand{\etal}{\emph{et al.}}

\newtheorem{problem}{Problem}
\newcommand{\spara}[1]{\smallskip\noindent{\bf{#1}}}

\newcommand{\bpara}[1]{\bigskip\noindent{\bf{#1}}}

\newcommand{\squishlist}{\begin{list}{$\bullet$}
  { \setlength{\itemsep}{0pt}
     \setlength{\parsep}{3pt}
     \setlength{\topsep}{3pt}
     \setlength{\partopsep}{0pt}
     \setlength{\leftmargin}{1.5em}
     \setlength{\labelwidth}{1em}
     \setlength{\labelsep}{0.5em} } }
\newcommand{\squishend}{
\end{list}  }

\begin{document}

\title{\Large Markov Chain Monitoring}

\author{Harshal A. Chaudhari\\ Computer Science \\ Boston University
\and
Michael Mathioudakis\\Computer Science \\ University of Helsinki
\and
Evimaria Terzi\\ Computer Science \\ Boston University
}
\date{}

\maketitle

% Copyright Statement
% When submitting your final paper to a SIAM proceedings, it is requested that you include 
% the appropriate copyright in the footer of the paper.  The copyright added should be 
% consistent with the copyright selected on the copyright form submitted with the paper.
% Please note that "20XX" should be changed to the year of the meeting.

% Depending on which copyright you agree to when you sign the copyright form, the copyright 
% can be changed to one of the following after commenting out the default copyright statement
% above.

% Default Copyright Statement
% \fancyfoot[R]{\footnotesize{\textbf{Copyright \textcopyright\ 2018 by SIAM\\
% Unauthorized reproduction of this article is prohibited}}}

% \fancyfoot[R]{\footnotesize{\textbf{Copyright \textcopyright\ 20XX\\
% Copyright for this paper is retained by authors}}}

%\fancyfoot[R]{\footnotesize{\textbf{Copyright \textcopyright\ 20XX\\
%Copyright retained by principal author's organization}}}

%\pagenumbering{arabic}
%\setcounter{page}{1}%Leave this line commented out.

\begin{abstract}
In networking applications, one often wishes to obtain estimates
about the number of objects at different parts of the network
(e.g., the number of cars at an intersection of a road network 
or the number of packets expected to reach a node in a computer network)
by monitoring the traffic in a small number of network nodes or edges. 
We formalize this task by defining the \mcproblem\ problem. 

Given an initial distribution of
items over the nodes of a {\markovchain}, we wish to
estimate the distribution of items at subsequent times. 
We do this by asking a limited number of queries that retrieve, for example,  
how many items transitioned to a specific node
or over a specific edge at a particular time.
We consider different types of queries, each defining a different variant of the {\mcproblem}.
For each variant, we design efficient algorithms for choosing the  queries that make our estimates 
as accurate as possible.
In our experiments with synthetic and real datasets we demonstrate the
efficiency and the efficacy of our algorithms in a variety of settings.

% -- but are
%constrained by the type and number of queries we are allowed to perform.
%The problem consists in choosing the queries so as to 
%make as accurate estimates as possible.
%% The problem finds natural application on {\textit{network traffic estimation}}
%% -- and leads to novel and intuitive definitions of {node} and {edge centrality}
%% on networks.
%
%In this paper, we define variants of the problem for different types 
%of queries, provide efficient algorithms to solve them, and
%showcase the performance of the algorithms on real datasets.
\end{abstract}
\section{Introduction}
In this paper, we introduce and solve the 
 \mcproblem\ problem. 
As multiple items move on a graph in \markovchain-fashion, 
our goal is to keep track of their distribution 
across the graph.
Towards that end, the {\mcproblem} 
problem
seeks to identify a limited number of nodes (or edges)
so that, once we know exactly 
how many items reside on (resp.\ traverse) them, our uncertainty 
of the distribution of items on the graph
is minimized.

%
%The situation is similar if, instead of nodes, 
%we monitor {\it edges} of the graph, to know
%the number of items that traverse them.
%The {\mcproblem} seeks
% to select  the set of nodes nodes or edges that one needs 
%monitor, so as to have an accurate estimate about the number of items 
%that reside on different nodes.

% This task can be viewed as the sampling of a 
% multi-item \markovchain, with the goal of preserving  as much 
% information as possible about the state of the system.
This problem finds applications in many network settings.
% One broad category of such settings is 
% including that of network traffic monitoring.
In a typical such setting, one wishes to have accurate estimates of
the number of items that reside at various
nodes of a network (e.g., vehicles that are at the intersections of a road network)
however faces constraints in how big a part of the network 
(e.g., intersections or road segments) they can monitor at any time.
For instance, in urban traffic networks 
different parts of the network are active at different 
times of the day or week. Therefore, actively measuring traffic 
at all points of the network would waste resources.
The task applies similarly on other types of networks including
computer or peer-to-peer networks where the items that propagate through the network are packages
or files.

To the best of our knowledge, we are the first to introduce and study the {\mcproblem} problem.
Nevertheless, this problem is similar
to tasks such as outbreak detection, 
sparsification of influence network, and the wider area
of node and edge {\it centrality} on graphs.
A major difference with previous work is that 
in \mcproblem\ the ``centrality" of nodes or edges
is defined not simply in terms of the
underlying graph structure, but also in terms of
the dynamic propagation of items through the network.

We make two assumptions in our problem definition.
The first is that there is a point in time when we have an accurate
estimate of the placement of all items on the graph.
The second is that we can monitor the subsequent placement
of items by issuing a predefined number of {\it monitoring operations} 
(i.e.,  real-time queries on the {\markovchain}).
We consider different types of operations:
ones that retrieve the number of items that reside on specific nodes;
and ones that retrieve the number of items that traverse specific edges. 
In the applications we consider, i.e., urban and computer traffic networks,
the queries correspond to placing measurement devices on particular 
nodes or edges of the network, and retrieving their measurements. 

Technically, different monitoring operations result in different
variants of  {\mcproblem}.
% e.g., monitoring of the items that pass through an edge leads to a different problem definition than monitoring of the items of the nodes. 
For each variant, we design efficient polynomial-time algorithms. 
For some of these algorithms we demonstrate that they are indeed
optimal, while for others we show that they perform extremely well in practice. 
Our experiments use a diverse set of datasets from urban networks and computer
networks to demonstrate the practical utility of our setting. 
For example, our experiments with data from the Hubway bike-sharing 
network of Boston, where the nodes are bike docks and the items that propagate
through the network are the bikes themselves, identify as candidates for monitoring
the stations that are close to the most busy Boston spots.

% MM: Can we skip the outline of the paper?
% The rest of the paper is structured as follows. We define the \mcproblem\
% and its variants in Section~\ref{sec:setting}.
% Subsequently, in sections~\ref{sec:nodes} to~\ref{sec:simplenodes},
% we provide efficient algorithms to solve each.
% % Up to that point,
% % the paper discusses the problem under the assumption that 
% % monitoring operations are performed at each step of the \markovchain.
% % Section~\ref{sec:infinity} provides some additional notes
% % for the case where monitoring operations are performed at larger
% % intervals which, in the extreme case, allow the \markovchain\ to converge
% % to steady state.
% Finally, in Section~\ref{sec:experiments} we showcase the performance of 
% algorithms and baselines on real data (traffic networks)
% and conclude with discussion about future work.

\section{Related work}

To the best of our knowledge, we are the first to introduce and study the
 \mcproblem\ problem. However, this problem can also be seen as 
 a graph centrality problem, where one seeks to 
identify $k$ ``central" nodes or edges to
intercept the movement of items on a graph. Therefore, our work is related
to existing work on graph centrality measures.

Graph centrality measures can be broadly cast into two categories: \emph{individual}
and \emph{group centrality} measures. Individual measures assign 
a score to the each node/edge. %and the best $k$ nodes/edges are the $k$ with the
%highest score. 
Group centrality measures assign scores to sets of nodes or edges. 
%In this categorization, our task of picking $k$ nodes or edges to monitor
%is a group centrality measure. 
Usually, computing group-centrality measures requires solving a
complex combinatorial problem. %and algorithms for computing
%the group with the highest centrality are much more involved.

Examples of individual centrality scores for nodes and edges are the 
\emph{Pagerank}~\cite{pagerank1999}, 
betweenness~\cite{brandes01faster,erdos15divide,riondato16fast} 
and current flow centralities~\cite{brandes05centrality}. 
Pagerank is one classical 
example of a centrality measure based on a \markovchain\
-- where the centrality of nodes is
quantified as the stationary distribution
of a \markovchain\ on the graph. 
Betweenness centrality and current flow centrality
assign high centrality score to nodes/edges that participate in one or many short
paths between all pairs of nodes in the input network. 
Although a {\markovchain} is used for Pagerank, its computation is very different 
than ours -- after all, Pagerank is an individual centrality measure, while our measures are group centralities.
% Similarly, betweenness and current flow centrality capture different intuitions and 
% raise different computational challenges from the ones we solve in this paper.

Group centrality measures that use a {\markovchain} model 
include the {\it Absorbing Random Walk Centrality} 
introduced by Mavroforakis {\etal}~\cite{Mavroforakis:2015}. 
In that work, the centrality of a set of nodes is defined
as the transient time of the \markovchain\ when 
the given nodes are ``absorbing". 
% Their task is then to find the 
% group of $k$ nodes with the highest global absorbing random walk centrality.
An absorbing {\markovchain} is also assumed by the work of Gionis {\etal}~\cite{gionis13maximizing}. 
In that setting, the authors aim to maximize the positive
opinion of the network with respect to a specific topic by picking $k$ nodes appropriately to endorse this positiveness via posts (e.g., in a social network).
Both these problems are different from ours because we do not consider 
absorbing random walks, but rather a simple {\markovchain}. 
Moreover, our objective to minimize uncertainty is different from the 
objectives of the above papers. 
% As a result the algorithmic challenges we 
% are addressing here are orthogonal to those addressed by their work. 

While \markovchain\ is a simple model that allows us to quantify 
centrality, different models are more appropriate
in other settings.
For example, Ishakian {\etal}~\cite{ishakian12framework} proposed an extension
of betweenness centrality to the group betweenness centrality. In that case, the goal
was to find a group of nodes such that many shortest paths pass through those nodes.
As another example, in epidimiology and information propagation
settings, the underlying process of interest (e.g., the spreading of a virus
or piece of information) is better modeled
as a random cascade.
In such a setting, the measure of 
{\it influence}~\cite{borgs2014maximizing, galhotra2015asim, goyal2011simpath,
Kempe:2003, ohsaka2014fast, Tang:2014} captures another notion of centrality, where
the centrality of a node is defined in
terms of the expected spread of a random cascade that starts at the
given node. In the same setting, two related tasks are those of outbreak 
detection~\cite{leskovec2007cost} and graph 
sparsification~\cite{Mathioudakis:2011}. In the former,
the task is to identify as central those nodes that would intercept 
early a set of observed cascades. In the latter, one seeks edges
that are central in explaining (from a model-sparsity point of view)
the observed cascades as well as possible.

In terms of categorization, the  the {\mcproblem} problem
is a group-centrality measure.
What distinguishes the setting of \mcproblem\
from all the above work is that 
centrality is defined in terms of a combination
of static structure (graph), 
dynamic structure (\markovchain), and real-time activity 
(current placement of items on the graph),  which is not the case in
the aforementioned works.

\section{Setting}
\label{sec:setting}
% In this section, 
In this section, we introduce our setting and notation.

\spara{Markov Chain:}
Consider a weighted directed graph 
$G=(V,E,p)$ with $|V|=n$ and $|E|=m$.
We use $\pi(v)\subseteq V$ and $\children(v)\subseteq V$ to 
denote the set of parent and child nodes of node $v$, respectively.
\begin{align*}
\pi(v) & = \{u\mid u\in V, e(u\rightarrow v)\in E\} \\
\kappa(v) & = \{u\mid u\in V, e(v\rightarrow u)\in E\}
\end{align*}
Edges $e(u\rightarrow v)\in E$ are associated with 
real-valued weights $p(u\rightarrow v)\in [0,1]$
such that $\forall u\in V$: 
$
  \sum_{v\in V} p(u\rightarrow v) = 1$.

These weights give rise to a \markovchain\
with transition matrix $\transition$ --
where $\transition(u,v) = p(u,v)$ denotes
the probability of a transition from node $u$ to node $v$. 
% Since 
% $\transition$ is the transition matrix of a \markovchain, it has the property
% that $\sum_{v\in V}\transition(u,v)=1$.
Moreover, we assume that a set of items
are distributed among the nodes of $G$. 
We use the row vector $\initial$ to denote the 
the {\it initial} number of items per node; 
that is, $\initial(u)$ is the number of items {initially}
at node $u$.
For the entirety of this paper, 
transition matrix $\transition$ and distribution $\initial$ are 
assumed known and part of the input.

\note[MM]{For consistency, let's use 'position' instead of 'location' of items.}
Consider now a single step of the \markovchain. During this step, each and
every item transitions from the node $u$ where it resides originally
to another node $v$, according to transition probabilities 
$\transition(u,v)$.
At the end of this step, items are 
redistributed among the nodes - and we are no longer certain about their
position. We use $\rvFinal$ to denote the random vector with the number
of items at each node after one transition.
The \emph{expected number of items} at each node is given by: 
%\[
$\final = E[\rvFinal] = \initial{\transition}$.
%\]

%% MM: Do we need this?
% More generally, 
% if the \markovchain\ takes $t$ steps, then the expected number of items
% per node at the end of these steps is given by:
% \[
% \final = \initial{\transition}^t.
% \]

\spara{Quantifying uncertainty:} 
We wish to have 
% an accurate estimate of the distribution of items over nodes $V$ 
% -- in other words, we wish to have 
good point estimates for values $\final(v)$.
We quantify the quality of estimates in terms of the variance
in the number of items on each node after the transition step.
Specifically, let us consider the number $\rvFinal(v)$ of items
on node $v$ after the transition step:
an item previously at node $u$ will transition to
node $v$ according to a Bernoulli distribution with 
success probability $\transition(u,v)$;
and since items transition from node to node independently from each other,
the variance in the number of items $\rvFinal(v)$ is
\begin{align}
\variance(\rvFinal(v))  =\ & \variance(\rvFinal(v) | \transition, \initial) \nonumber \\
 =\ & \sum_{u\in V}\initial(u)\transition(u,v)\left(1-\transition(u,v)\right).
\end{align}
To obtain an aggregate measure of uncertainty $\uncertainty_0$,
we opt to sum the aforementioned quantity over all nodes.
% with simple summation, as follows:
\begin{align}
\uncertainty_0 & = \sum_{v\in V}{\variance(\rvFinal(v))} \nonumber  \\
 %& = \sum_{v\in V}\sum_{u\in V}\initial(u)\transition(u,v)\left(1-\transition(u,v)\right) \nonumber \\
 & = \sum_{u\in V}\initial(u)\sum_{v\in V}\transition(u,v)\left(1-\transition(u,v)\right) 
\label{eq:uncertainty}
\end{align}

\spara{Monitoring:} 
Given the transition matrix $\transition$ and the 
initial distribution of
items $\initial$, we estimate the distribution of items 
$\rvFinal$ after
one transition step, with
the uncertainty given in Equation~\eqref{eq:uncertainty}.
After one transition step,
we are allowed to 
retrieve information about the position of the items
and thus reduce uncertainty.
We do this by performing ``monitoring operations", i.e. queries on the
position of items on the \markovchain.
These operations are  of the following types:
%\begin{itemize}
\squishlist
	\item{\bf \parentstransitions} Retrieve the number of items that transitioned 
  to node $v$ from each  $u\in\pi(v)$;
  \item{\bf \nodeitems} Retrieve the number of items that reside on node 
  $v$ after the transition step;
	\item{\bf \edgetransitions} Retrieve the number of items that transitioned 
  from node $u$ to node $v$;
	\item{\bf \childrentransitions} Retrieve the number of items that 
  transitioned from node $u$ to each child $v\in \children(u)$.
%\end{itemize}
\squishend

% From the above four types of monitoring operations, the last one (i.e., {\childrentransitions}) is both a bit 
% unintuitive and leads to trivial combinatorial problems. Thus, we omit them from the rest of the discussion.
%%%%%%

\bpara{Expected uncertainty} 
Once we retrieve the answer $\answer$ to a set of monitoring operations,
we have more information about the positioning of items over nodes $V$ 
-- and thus an updated (and non-increased) uncertainty 
$\uncertainty(\answer) = \sum_{v\in V}{\variance(\rvFinal(v) | \answer)}$.
In the setting we consider, however, the challenge we face is {\it not} to compute
the uncertainty {\it given} the information retrieved via a monitoring 
operation, but rather to {\it select the monitoring operations} 
so that the uncertainty we will face {\it after} retrieving $\answer$ is
minimized {\it in expectation}. Therefore, the quantity of interest
is that of \expecteduncertainty\ for a set of operations
that we choose to perform, expressed as 
$E[\sum_{v\in V}{\variance(\rvFinal(v) | \answer)}]$. %,
% with expectation taken
% over the possible transitions according to \transition.

Let us now assume we have chosen to perform a set of operations of 
either one of the aforementioned types. In what follows, we provide 
formulas for the \expecteduncertainty\
in each case.

\spara{Expected uncertainty under \parentstransitions :} 
We 
perform monitoring operations for
a subset $S\subseteq V$ of nodes --
and obtain an answer 
$A_{_\shortparentstransitions}(S) = \{n_{uv}; v\in S,\ e(u\rightarrow v)\in E\}$,
where $n_{uv}$ is the number of transitions to $v$ from its parent node $u$. 
The expected value $\objective_{_\shortparentstransitions}(S)$ of the uncertainty 
$\uncertainty(A_{_\shortparentstransitions}(S))$ 
after these operations is given by
\begin{align}
\label{eq:nodevariance}
\objective_{_\shortparentstransitions}(S) = & E[\uncertainty(A_{_\shortparentstransitions}(S))] \nonumber \\ 
= & \sum_{u\in V}\initial^{\prime}(u)\sum_{v\in V\setminus S}\transition^{\prime}(u,v)\left(1-\transition^{\prime}(u,v)\right) 
\end{align}
where 
% expectation is taken over the probability of
% possible transitions and thus
% possible answers $A_{_\shortparentstransitions}$ -- and 
% the following notation is used.
\begin{align}
\rho(u,S) & = \sum_{v\in S}\transition(u,v) \nonumber \\
 \initial^{\prime}(u) & = \initial(u)\left(1-\rho(u,S)\right) \label{eq:adjustedx}\\
 \transition^{\prime}(u,v) & = \frac{\transition(u,v)}{1-\rho(u,S)}. \label{eq:adjustedP}
\end{align}
Intuitively, $\initial^{\prime}(u)$ expresses the expected number
of items that transition from $u$ to nodes $v$ {\it other than} those in $S$;
and $\transition^{\prime}(u,v)$ expresses the probability an item transitions
from $u$ to $v$ {\it given} that it does {\it not} transition to those in $S$.
We see, then, that Equation~\eqref{eq:nodevariance} has the same form as 
Equation~\eqref{eq:uncertainty} but is evaluated on adjusted values of
$\initial$ and $\transition$, to take into account the
information we obtain via $A_{_\shortparentstransitions}$.

Note that the expected uncertainty after
the monitoring operations is no larger than $\uncertainty_0$.
The following lemma states that the information retrieved via
\parentstransitions\ monitoring operations decrease the expected uncertainty
about the positioning of items after the transition step.

\begin{lemma}
\[
\objective_{_\shortparentstransitions}(S) \leq \uncertainty_0
\]
\label{lemma:decreased_uncertainty_pt}
\end{lemma}
\begin{proof}
Indeed, using Equations~\eqref{eq:adjustedx} and~\eqref{eq:adjustedP} to expand 
Equation~\eqref{eq:nodevariance} we have
\begin{align}
\objective_{_\shortparentstransitions}(S) =\ & \sum_{u\in V}\initial^{\prime}(u)\sum_{v\in V\setminus S}\transition^{\prime}(u,v)\left(1-\transition^{\prime}(u,v)\right) \nonumber \\
  =\ & \sum_{u\in V}\initial(u)\sum_{v\in V\setminus S}\transition(u,v)\left(1-\frac{\transition(u,v)}{1-\rho(u,S)}\right) \nonumber  \\
  \leq\ & \sum_{u\in V}\initial(u)\sum_{v\in V\setminus S}\transition(u,v)\left(1-\transition(u,v)\right) \nonumber \\
  =\ & \objective_{_\shortparentstransitions}(\emptyset). \nonumber
\end{align}
\end{proof}

\spara{Expected uncertainty under \nodeitems:} 
We 
perform monitoring operations for
a subset $S\subseteq V$ of the nodes --
and obtain an answer 
$A_{_\shortnodeitems}(S) 
= \{n_v; v\in S\}$, where $n_v$ is the number of items at node $v$
after the transition.
For an instance of answer $A_{_\shortnodeitems}(S)$,
let also $A_{_\shortparentstransitions}(S) = \{n_{uv}; v\in S, 
  e(u\rightarrow v)\in E\}$ 
be an answer for \parentstransitions\ on the same set $S$ of nodes
-- which by definition
is {\it consistent with} $A_{_\shortnodeitems}(S)$, in the sense that
\begin{equation}
n_v = \sum_{u\in V} n_{uv}, \forall v\in S.
\end{equation}
%It can be shown that the \expecteduncertainty\ is equal for the two cases.
%That is: 
%$\objective_{_\shortparentstransitions}(S) = \objective_{_\shortnodeitems}(S)$.
%\footnote{For a proof, see Supplementary Material, 
%Theorem~\ref{theorem:node-equivalence}.}.
The following theorem states that, for the same set of monitored nodes,
\parentstransitions\ and \nodeitems\ lead to the same value of the objective 
function.
\begin{theorem}
\label{theorem:node-equivalence}
$\objective_{_\shortnodeitems}(S) = \objective_{_\shortparentstransitions}(S)$
\begin{proof}
We express $\uncertainty(A_{_\shortnodeitems})$ in terms
of $\uncertainty(A_{_\shortparentstransitions})$ as follows.
(We write `c.w.' for `consistent with').
\begin{align}
\uncertainty(A_{_\shortnodeitems})
 = \sum_{A_{_\shortparentstransitions}
  \ \text{c.w.}\; A_{_\shortnodeitems}
  } \uncertainty(A_{_\shortparentstransitions})
  \prob(A_{_\shortparentstransitions} | A_{_\shortnodeitems})
\end{align}
We can now use the above equation to
express the expected uncertainty $\objective_{_\shortnodeitems}(S)$
in terms of $\objective_{_\shortparentstransitions}(S)$ as:
\begin{align}
\objective_{_\shortnodeitems}(S) & =  E[\uncertainty(A_{\shortnodeitems})] \nonumber \\
& = \sum_{A_{_\shortnodeitems}}\uncertainty(A_{\shortnodeitems}) \prob(A_{_\shortnodeitems}) = \nonumber \\
& = \sum_{A_{_\shortnodeitems}}\;\;\sum_{A_{_\shortparentstransitions}
   \text{c.w.} A_{_\shortnodeitems}
  } \uncertainty(A_{_\shortparentstransitions})
  \prob(A_{_\shortparentstransitions} | A_{_\shortnodeitems}) \prob(A_{_\shortnodeitems}) \nonumber \\
& = \sum_{A_{_\shortnodeitems}}\;\;\sum_{A_{_\shortparentstransitions}
  \ \text{c.w.}\ A_{_\shortnodeitems}
  } \uncertainty(A_{_\shortparentstransitions})
  \prob(A_{_\shortparentstransitions}, A_{_\shortnodeitems}) \nonumber \\
& = \sum_{A_{_\shortnodeitems}}\;\;\sum_{A_{_\shortparentstransitions}
  \ \text{c.w.}\ A_{_\shortnodeitems}
  } \uncertainty(A_{_\shortparentstransitions})
  \prob(A_{_\shortparentstransitions}) \nonumber \\
& = \sum_{A_{_\shortparentstransitions}} 
  \uncertainty(A_{_\shortparentstransitions})
  \cdot \prob(A_{_\shortparentstransitions}) \nonumber \\
& = \objective_{_\shortparentstransitions}(S) ,\nonumber
\end{align}
which concludes the proof.
\end{proof}
\end{theorem}

%-- i.e., whether we are given $A_{_\shortnodeitems}(S)$ 
%or $A_{_\shortparentstransitions}(S)$ -- expressed by
%Equation~\eqref{eq:nodevariance}

\spara{Expected uncertainty under \edgetransitions:}
We perform monitoring operations for
a subset $D\subseteq E$ of the edges --
and obtain an answer 
$A_{_\shortedgetransitions} = A_{_\shortedgetransitions}(D) = \{n_e; e\in D\}$,
where $n_e$ is the number of transitions over edge $e$. 
The expected value $\objective_{_\shortedgetransitions}(D)$ of the uncertainty 
$\uncertainty(A_{_\shortedgetransitions}(D))$ 
after these operations is given by
\begin{align}
\objective_{_\shortedgetransitions}(D) = & E[\uncertainty(A_{_\shortedgetransitions})] \nonumber \\
= & \sum_{u\in V}\initial^{\prime\prime}(u)\sum_{e(u\rightarrow v)\in E\setminus D}\transition^{\prime\prime}(u,v)\left(1-\transition^{\prime\prime}(u,v)\right) \label{eq:edgetransitions}
\end{align}
where 
\begin{align}
\rho(u,D) & =\ \sum_{e(u\rightarrow v)\in D}\transition(u,v) \nonumber\\
 \initial^{\prime\prime}(u) & =\ \initial(u)\left(1-\rho(u,D)\right) \label{eq:etinit}\\
 \transition^{\prime\prime}(u,v) & =\ \frac{\transition(u,v)}{1-\rho(u,D)} \label{eq:ettransit}
\end{align}
Similar to \parentstransitions\ and \nodeitems, expected uncertainty
$\objective_{_\shortedgetransitions}(S)$ is no greater than
$\uncertainty_0$.

\spara{Expected uncertainty under \childrentransitions:}
Assume that we perform monitoring operations for
a subset $S\subseteq V$ of the nodes --
and obtain an answer
$A_{_\shortchildrentransitions}(S) 
= \{n_{uv}; u\in S, (u, v)\in E\}$, where $n_{uv}$ is the number of items that 
transition over edge $(u,v)$.
The expected value $\objective_{_\shortchildrentransitions}(S)$ of the uncertainty 
$\uncertainty(A_{_\shortchildrentransitions}(S))$ 
after these operations is given by
\begin{align}
\label{eq:shortchildrentransitions}
\objective_{_\shortchildrentransitions}(S) & = E[\uncertainty(A_{_\shortchildrentransitions})] = \nonumber \\
& = \sum_{u\in V-S}\initial(u)\sum_{v\in V}\transition(u,v)\left(1-\transition(u,v)\right).
\end{align}
Notice that this quantity is no greater than
$\uncertainty_0$ (Equation~\eqref{eq:uncertainty}), as the outer summation
is performed for a subset of nodes.

%!TEX root = arxiv.tex
\section{Problem Definition}
\label{section:problems}

The general problem of \mcproblem\ is to select the
appropriate monitoring operations to reduce
the expected uncertainty after they are performed.
Stated formally:

\begin{problem}[\mcproblem]
Given a transition matrix $\transition$ and an initial distribution of items $\initial$,
select a set of up to $k$ monitoring operations 
to minimize the expected uncertainty $\objective$. 
% $E[\uncertainty]$
\end{problem}

We study variants of the problem --
each defined for a specific type of monitoring operation.
 For simplicity, we refer to these problems
 with the same name as that of the operation type:
\variant{\parentstransitions},
\variant{\nodeitems},
\variant{\childrentransitions}, 
and \variant{\edgetransitions}.

Furthermore, as we saw in Section~\ref{sec:setting},
variants
\variant{\parentstransitions}\ and
\variant{\nodeitems} are equivalent: for the same set of nodes,
operations of the first type
reduce expected uncertainty as much as the second.
Therefore, in what follows,
we treat only the variant of \variant{\nodeitems}, as  our claims apply directly to
\variant{\parentstransitions}\ as well.

% \todo[MM]{One issue we have not addressed is how we go from 
% $A_{_\shortnodeitems}(S)$ to $A_{_\shortparentstransitions}(S)$.
% Does it affect our claims?}

\section{The {\variant{\nodeitems}} problem}
\label{sec:nodes}

In this section, we provide the formal problem definition 
of the {\variant{\nodeitems}} problem variant and describe a greedy 
polynomial-time algorithm for solving it.
\begin{problem}[{\variant{\nodeitems}}]
Given $G=(V,E)$, transition matrix $\transition$, initial distribution of items to nodes
$\initial$ and integer $k$, find
$S\subseteq V$ such that $|S|=k$ such that 
$\objective_\shortnodeitems\left(S\right)$ is minimized.
\label{problem:nodes-variant}
\end{problem}
A brute-force way to solve Problem~\ref{problem:nodes-variant}
would be to evaluate the objective function over all node-sets of size $k$.
Obviously such an algorithm is infeasible -- and we thus
study a natural greedy algorithm for the problem, namely \nodegreedy.

\spara{The {\nodegreedy} algorithm:} 
This is a greedy algorithm that performs $k$ iterations; at each iteration,
it adds one more 
node in the solution.
If $S_t$ is the solution at iteration $t$, 
then solution $S_{t+1}$ is constructed by finding the 
node $u\in V\setminus S_t$ such that:
\begin{equation}\label{eq:nodegreedy}
v^\ast=\argmin_{v\in V\setminus S_t}\objective_\shortnodeitems\left(S_{t}\cup\{v\}\right).
\end{equation}
Although in the majority of our experiments that compare the brute-force solutions with those
of {\nodegreedy} the two solutions were identical, we identified some contrived instances 
for which this was not the case. Thus, {\nodegreedy} is not an optimal algorithm for
Problem~\ref{problem:nodes-variant}.

\spara{Running time} 
{\nodegreedy} evaluates Equation~\eqref{eq:nodegreedy} at each iteration. 
A naive implementation of this would require computing Equation~\eqref{eq:nodevariance}
$O(|V|)$ times per iteration, each time using $O(|V|^2)$ numerical operations.
As a first improvement, we avoid the full double summation over $V$
via a summation over edges $E$,
\begin{align}
	\objective_{_\shortparentstransitions}(S) = & \sum_{u\in V}\initial^{\prime}(u)\sum_{v\in V\setminus S}\transition^{\prime}(u,v)\left(1-\transition^{\prime}(u,v)\right) \nonumber \\ 
	= &  \sum_{(u,v)\in E, v\in V\setminus S} 
	\initial^{\prime}(u)\transition^{\prime}(u,v)\left(1-\transition^{\prime}(u,v)\right),
\end{align}
that involves $O(k|V||E|)$ numerical operations.

% Clearly, the above running time would make {\nodegreedy}
% infeasible to run even for small-size datasets. 
We can further speed-up the algorithm if 
we re-use at each step the computations 
done in the previous one.
To see how, let $S_t$ (resp.\ $S_{t+1}$) 
be the solution we construct after
$t$ (resp.\ $(t+1)$) iterations and 
let $v^\ast$ be the node such that
$S_{t+1}=S_t\cup v^\ast$. 
Then, for any $u\in V$ we have $
\rho(u,S_t) = \sum_{v\in S_t}\transition(u,v)
$, and therefore
\begin{eqnarray}\label{eq:rho}
% \rho(u,S_{t+1}) & = &\sum_{v\in S_{t+1}}\transition(u,v)\\ 
% & = & \sum_{v\in S_{t}}\transition(u,v) + \transition(u,v^\ast)\nonumber\\
% & = & \rho(u,S_{t})+ \transition(u,v^\ast)\nonumber
\rho(u,S_{t+1}) & = \rho(u,S_{t})+ \transition(u,v^\ast).
\end{eqnarray}
Moreover, for any $S\subseteq V$ let
\begin{align}
B(u,S) & = \sum_{(u,v)\in E\ s.t.\ v\in V\setminus S}\transition^{\prime}(u,v)\left(1-\transition^{\prime}(u,v)\right) \\
& = \sum_{v\in V\setminus S}\frac{\transition(u,v)}{1-\rho(u,S)}\left(1-\frac{\transition(u,v)}{1-\rho(u,S)}\right)\nonumber .	
\end{align}
We can then express $B(u,S_{t+1})$ in terms of
$B(u, S_{t})$:
\begin{eqnarray}\label{eq:B}
\lefteqn{
B(u,S_{t+1})  =} \nonumber \\
& & B(u,S_t)-2\transition(u,v^\ast)\left(1-\rho(u,S_t)-\transition(u,v^\ast)\right). 
\end{eqnarray}
Finally, using Equations~\eqref{eq:rho} and~\eqref{eq:B} 
and algebraic manipulations,
we can express $\objective_{\shortnodeitems}(S_{t+1})$ as follows:
\begin{eqnarray}\label{eq:noderewrite}
\lefteqn{
\objective_{\shortnodeitems}(S_{t+1}) = }\\
&&\sum_{u\in V}\initial(u)\left(\frac{B(u,S_t)}{1-\rho(u,S_{t+1})}-2\transition(u,v^\ast)\right)\nonumber
\end{eqnarray}
Thus, if we store $B(u,S_t)$ and $\rho(u,S_t)$ at iteration $t$, 
then evaluating 
Equation~\eqref{eq:noderewrite} at iteration $t+1$
takes only $O(|V|)$ numerical operations.

For all iterations but the first one,
the above sequence of rewrites enables us to achieve 
a speedup from $O(|V||E|)$ to $O(|V|^2)$ numerical operations
per iteration.
For the first iteration, initializing
the auxiliary quantities $B(u, \emptyset)$, $u\in V$, still takes
$O(|E|)$.
With this book-keeping, the running time of {\nodegreedy} is  reduced
from $O(k|V||E|)$ to $O(|E| + k|V|^2) = O(k|V|^2|)$.
Note also that \nodegreedy\ is amenable to parallelization,
as, given the auxiliary quantities from the previous step, 
we can compute the objective function independently for each 
candidate node.

%!TEX root = arxiv.tex
\section{The \variant{\edgetransitions} problem}
\label{sec:edges}

Whereas \variant{\nodeitems} (Problem~\ref{problem:nodes-variant}) seeks $k$ 
nodes to optimize expected uncertainty, \variant{\edgetransitions} seeks
$k$ edges.
\begin{problem}[{\edgeproblem}]
Given $G=(V,E)$, transition matrix $\transition$, 
initial distribution of items to nodes
$\initial$ and integer $k$, find
$S\subseteq V$ such that $|S|=k$ such that 
$\objective_\shortedgetransitions\left(S\right)$
(Equation~\eqref{eq:edgetransitions}) is minimized.
\label{problem:edges-variant}
\end{problem}
We provide two polynomial-time algorithms to solve
the problem, namely \edgeDP\ and \edgegreedy.
For the former, we can also prove that it is optimal, and thus Problem~\ref{problem:edges-variant}
is solvable in polynomial time.

\spara{The {\edgeDP} algorithm:}
{\edgeDP} is a dynamic-programming algorithm that 
selects edges in two steps: first, it sorts the {\it outgoing} 
edges of
each node in decreasing order of transition probability, thus
creating $|V| = n$ corresponding lists; secondly, it combines top edges from
each list to select a total of $k$ edges.
%The algorithm is polynomial and optimal.

In more detail, let $SOL_i(k)$ be the cost
of an optimal solution for the special case of
a budget of $k$ edges allocated among outgoing edges of nodes
$V_{i:} = \{i, i+1, \ldots, n\}$.
According to this notational convention,
the cost of an optimal solution $D_{opt}$ for the problem is given by 
$SOL_1(k)$.
%  -- 
% and we define $SOL_i(k) = \infty$ for $k < 0$ and $SOL_i(k) = 0$ 
% for $i > \|V\| = n$.
Moreover, considering Equation~\eqref{eq:edgetransitions},
let $\objective_i$ be the function that corresponds to
the (outer) summation term for node $i$
\begin{equation}
	\objective_i(D) = \initial^{\prime\prime}(i)\sum_{v\in V\setminus S}\transition^{\prime\prime}(i,v)\left(1-\transition^{\prime\prime}(i,v)\right)
\end{equation}
(under the auxiliary definitions of Equations \eqref{eq:etinit}
and \eqref{eq:ettransit})
and $ISOL_i(m)$ its optimal value when $D$ contains no more than
$m\leq k$ outgoing edges from node $i$.
Let also $D_i^m$ be a subset of $k$ outgoing edges of $i$
with the highest transition
probabilities. 
It can be shown that the optimal value
for $\objective_i(D)$ is achieved for 
the edges $D_i^k$ with {\it highest} transition 
probability.
The following lemma states that, with choice restricted 
among the outgoing edges of a node, the optimal
objective value in the \edgetransitions\ setting is obtained for
the edges of highest transition probability.
\begin{lemma}
\begin{equation}
ISOL_i(m) =  \objective_i(D_i^m)	
\end{equation}
\begin{proof}
The optimization function is proportional to the following quantity:
\begin{equation}
f(E) \propto (\sum_{i\in D_u(E)} p_i) - {\sum_{i\in D_u(E)} p_i^2} / {(\sum_{i\in D_u(E)} p_i)}
\end{equation}
where $D_u(E)$ are the remaining (i.e., non-queried) outgoing edges of
parent-node $u$.

Consider two sets of edges $E_0$, $E_1$ $\subseteq O(u)$ of the same size,
all outgoing from a
single parent-node $u$, that differ only at one element.
The probabilities of the corresponding sets of {\bf remaining} edges
are:
\begin{equation}
 D_u(E_0): \{p_0\} \cup C;\;\;\;\;D_u(E_1): \{p_1\} \cup C
\end{equation}
where $p_0, p_1\not\in C$, $p_0 \leq p_1$.

Let $S = \sum_{i\in C} p_i$ and $SS = \sum_{i\in C} p_i^2$.
We take the difference of the optimization functions for the two sets $E_0$
and $E_1$.
\begin{align*}
	f(E_0) - f(E_1) \propto\ & p_0 - p_1 -  \frac{\sum_{i\in D_u(E_0)}{p_i^2}}{\sum_{i\in D_u(E_0)}{p_i}}  
		+  \frac{\sum_{i\in D_u(E_1)}{p_i^2}}{\sum_{i\in D_u(E_1)} {p_i}} \\
	 = & - (p_1 - p_0)\frac{SS + S^2}{(S + p_0)(S + p_1)} \leq 0.
\end{align*}
The above shows that selecting the set of edges so that the remaining edges
are associated with smaller probabilities leads to lower (better) values of the
optimization function.
\end{proof}
\label{lemma:singlenodeoptimality}
\end{lemma}
Having the outgoing edges of $i$ sorted by transition probability,
we can compute $ISOL_i(m)$ for all $m = 0\ldots k$.

The dynamic programming equation is: 
\begin{equation}
	SOL_i(k) = \argmin_{0\leq m\leq k}\{ISOL_i(m) + SOL_{i+1}(k - m)\}
	\label{eq:dptop}
\end{equation}
\edgeDP\ essentially computes and keeps in memory 
$\|V\| \times (k+1)$ values according to Eq.\eqref{eq:dptop}.

% Given the above, we have the following 
% result\footnote{For a proof sketch of this lemma see Supplementary Material, 
% Theorem~\ref{theorem:edgeDP}.}.
\begin{lemma}
The {\edgeDP} algorithm is optimal for the {\edgeproblem} problem.
\begin{proof}
The proof follows from Lemma~\ref{lemma:singlenodeoptimality}
and by construction of the dynamic programming algorithm
(Equation~\eqref{eq:dptop}).
\end{proof}
\end{lemma}

% \begin{algorithm}
% \LinesNumbered
%  \KwIn{k}
%  \KwOut{SOL: Dynamic programming array}
%  initialize empty array $SOL_{\|V\| \times (k+1)}$;\\
%  \For{i = $\|V\|$..1}{
%   \For{$k'$ = 0:k}{
%     SOL[$i$, $k'$] := $\argmin_{0\leq k_i\leq k'}\{ISOL_i(k_i) + SOL[i+1, k' - k_i]\}\}$
%    }
%  }
%  \Return SOL;
%  \caption{Dynamic programming algorithm for the \edgetransitions\ variant.}
%  \label{alg:edgeDP}
% \end{algorithm}

\emph{Running time:} \edgeDP\ 
computes $k\times |V|$ values. For each value to be
computed, up to $O(k)$ numerical operations are performed. Therefore,
\edgeDP\ runs in $O(k^2|V|)$ operations. Backtracking to retrieve the
optimal solution requires at most equal number of steps, so it does not
increase the asymptotic running time.

\spara{The {\edgegreedy} algorithm:}
\edgegreedy\ is a natural greedy algorithm that selects
$k$ edges in an equal number of steps, in each step selecting
one more edge to minimize $\objective_{\shortedgetransitions}$.
% \edgegreedy\ is described in Algorithm~\ref{alg:edgegreedy}.

% \begin{algorithm}
% \LinesNumbered
%  \KwIn{k}
%  \KwOut{$ResultEdges$: Set of selected edges}
%  $ResultEdges$ = $\{\}$ \\
%  \For{$i = 1 \cdots k$}{
%  	Select $e\in E$ := 
%  		$\argmin \objective_{\shortedgetransitions}(ResultEdges \cup \{e\})$ \\
%  	$ResultEdges$ := $ResultEdges$ $\cup$ \{e\}; \\
%  	E := E $\setminus$ \{e\}; \\
%  }
%  \Return $ResultEdges$;
%  \caption{Greedy algorithm for the \edgetransitions\ variant.}
%  \label{alg:edgegreedy}
% \end{algorithm}

%Unfortunately, we have not been able to prove or dispove optimality
%for \edgegreedy. 

In all our experiments the performance of {\edgegreedy} is the same as the performance of
the optimal {\edgeDP} algorithm. However, we do not have a proof that the former is also optimal. We leave this as a problem for future work.

\emph{Running time:} Following Equation \eqref{eq:edgetransitions},
to select $k$ edges,
\edgegreedy\ %as defined in Algorithm~\ref{alg:edgegreedy}
invokes up to $k\times O(|E|)$ evaluations of
$\objective_{\shortedgetransitions}$.
% ; and each evaluation of
% $\objective_{\shortedgetransitions}$ involves $O(|E|)$ numerical operations.
As we discussed for \nodegreedy, if the evaluation of the 
objective function is naively implemented with a double summation,
the running time of \edgegreedy\ is 
$O(k|E||V|^2)$ numerical operations.
If the objective function is implemented as a summation over edges,
the running time improves to $O(k|E|^2)$.
Furthermore, following the observations similar to those we saw for {\nodegreedy}, the
running time of \edgegreedy\ becomes $O(|E| + k|E|) = O(k|E|)$.

We notice that \edgeDP\ has better
performance than \edgegreedy\ for dense graphs ($|E|\approxeq |V|^2$) 
and small $k$. Moreover, as with \nodegreedy, \edgegreedy\ is 
amenable to parallelization - the new value of the objective function
can be computed in independently for each edge that's considered for
selection.

% We have the following theorem for \edgegreedy.
% \begin{theorem}
% \edgegreedy\ is optimal for the \edgetransitions\ variant of 
% the \mcproblem\ problem.
% \label{theorem:edgegreedy}
% \end{theorem}
% The proof for Theorem~\ref{theorem:edgegreedy} mirrors the proof for 
% Theorem~\ref{theorem:nodegreedy}, adjusted to the selection of edges instead
% of nodes - and is omitted for brevity.

%\input{simplenodes}
% \input{centrality}
%!TEX root = arxiv.tex
\section{Experiments}
\label{sec:experiments}

% \textbf{Graphs}

% \begin{itemize}
% 	\item ASS1 \\
% 	Nodes: 3015\\
% 	Edges: 5539

% 	\item GEO \\
% 	Nodes: 1000 \\
% 	Unit square with random nodes. Nodes within 0.01 distance are connected.

% 	\item GRID \\
% 	Nodes: 1000 \\
% 	All nodes except ones on border have out-degree and in-degree 4

% 	\item BA \\
% 	Nodes: 1000 \\
% 	preferential attachment parameter: 10, 5, 3. Results are for 10, but they don't change qualitatively for any other parameter.
% \end{itemize}
% \textbf{Item distributions}
% \begin{itemize}
% 	\item Uniform \\
% 	Items per node is total items / total nodes

% 	\item Direct \\
% 	Items directly proportional to out-degreee

% 	\item Inverse \\
% 	Items inversely proportional to out-degree. (directly proportional to number of nodes that are not the neighbors of given node.) \\
% 	Very similar to Uniform in sparse graphs.

% 	\item Ego \\
% 	Randomly chosen node. 70\% items distributed in a 2-radius sphere around randomly chosen node. 30\% outside. \\
% 	All results are averaged over 10 runs.

% \end{itemize}
In this section, we describe the results of our experimental evaluation using
real and synthetic data. The results demonstrate that our methods
perform better than other baseline methods with respect to our objective function. 
Moreover, using the bike-sharing network of Boston,
we provide anecdotal evidence that our methods pick meaningful nodes
to monitor.

\subsection{Experimental setup}

Let us first describe the experimental setup, i.e., the 
datasets and baseline algorithms used for evaluation.

\spara{Graph datasets:} We use the following graphs to 
define {\markovchain}s for our experiments.

\emph{\autonomoussystems} is a graph that contains information about 
traffic  between \textit{Autonomous Systems} of the Internet. 
%(An Autonomous System (AS)
%can be thought of as a set of IP addresses, 
%typically characterized by their common prefix 
%and belonging to a single internet provider or large organization). 
The dataset was retrieved through the Stanford Large Network Dataset 
Collection (SNAP)~\cite{snapnets}.
We experimented with three snapshots of the %Autonomous Systems
communication graphs between years 1997 and 2000. 
Here we demonstrate results for one of the snapshots (1997-2000), 
as we did not find significant difference among them.

The {\autonomoussystems} graph contains one node for each 
AS. Moreover, for every pair of nodes between which there is traffic according
to the dataset, we place two directed edges between the nodes, one in
each direction. The resulting graph network contains 3015 nodes and 
11078 edges. To create an instance of the transition matrix,
we assign equal probabilities to the outgoing edges of each node.

\emph{{\grid} graphs:} The {\grid} graphs are 
planar, bi-directed grid graphs, 
where each node has in- and out-degree $4$ 
(with the exception of border nodes). The graph used in our experiments contains a total of 1000 nodes
in form of a $100 \times 10$ grid.

\emph{{\geo} graphs:} The {\geo} graphs are 
bi-directed geographic graphs. They are generated as follows:
1000 nodes are placed randomly within a unit square 
on the two-dimensional euclidean plane.
Subsequently, pairs of nodes are connected with directed edges in both
directions if their euclidean distance is below a pre-defined threshold
$ds = 0.01$.

\emph{{\ba} graphs}: The {\ba} graphs are  generated according to the
Barabasi-Albert model. According to the model,
nodes are added to the graph incrementally one after the other, 
each of them with outgoing edges to $m$ existing nodes selected
via preferential attachment. Here we show results for a graph with 1000 nodes and $m=3$, but the results
were similar for values $m=5,10$. 

Similar to the methodology of Gionis {\etal}~\cite{gionis2015bump}, 
the {\grid}, {\geo} and {\ba} graphs provide us 
with different varieties of synthetic graphs to explore the performance 
of our methods.

\spara{Item distributions:}
For each aforementioned graph, we generate an initial distribution
of items \initial\ according to one of the following four schemes.

%\begin{itemize}
\squishlist
\item \ego. Items are assigned in two steps. Firstly, one node
	is selected uniformly at random among all nodes. Secondly,
	$70\%$ of items are assigned randomly on the neighbors of the selected node
	(including the selected node itself). Finally, the remaining items
	are distributed randomly to the nodes outside the neighborhood of the 
	selected node.
\item \uniform. Each node is assigned the same number of items.
\item \direct. The number of items on each node is directly proportional to its
	out-degree. Note that items are distributed in a deterministic manner.
\item \inverse. The number of items on each node is assigned deterministically to
	be inversely proportional to its out-degree.
\squishend
%\end{itemize}

Now each graph described above is combined with 
each item-distribution scheme. 
As a result, we obtain datasets of the form {\tt G-X}, where 
{\tt G} is any of {\autonomoussystems}, {\grid},  {\geo} and {\ba} 
and {\tt X} is any of the {\ego}, {\uniform}, {\direct} and {\inverse}.
For simplicity, for the datasets that are generated randomly, 
we perform experiments over a single fixed instantiation.

\spara{The {\hubway} dataset:} Hubway
is a bike-sharing system in the Boston metro area, with a fleet of 
over 1000 bikes and over 100 hubway stations where users can pick up or 
drop off bikes at.
Every time a user picks up a bike from a Hubway station, 
the system records basic
information about the trip, such as the pick-up and drop-off
station, and the corresponding pick-up and drop-off times. 
Moreover, the data contain the number of available bikes
at each Hubway station every minute. 
The dataset was made publicly available by Hubway for the purposes of
its Data Visualization 
Challenge\footnote{http://hubwaydatachallenge.org/}.

Using the dataset, we create instances of the problems we consider as follows.
Firstly, we create a complete graph by representing each station with one
node in the graph, and considering all possible edges between them.
Subsequently, we consider a time interval $(t_s, t_e)$ and the bikes that
are located at each station (node). 
Representing bikes as items in our setting, we assign a transition 
probability $\transition(u,v)$ between nodes $u$ and $v$ by considering the 
total number $n_u$ of bikes at station $u$ at start time $t_s$ and, among these bikes,
the number $n_{uv}$ of them that were located at station $v$ at end time $t_e$.
We then set $\transition(u,v) = n_{uv}/n_u$ and ignore edges with zero
transition probability.

We experimented with a large number of such instances for different intervals
$(t_s, t_e)$, with a moderate length of $2$ hours, to capture real-life
transitions from one node to another. For the experiments presented in the 
paper, we use a fixed instance for the interval between 10am and 12pm on
April 1st, 2012. In this interval, we consider 61 stations with at least one trip
starting or ending at each. We refer to the dataset so constructed as the {\hubway}
dataset.

\spara{Baseline algorithms:}
In order to assess the performance of our proposed algorithms for the
{\nodeproblem} and {\edgeproblem} problem variants,
we compare it to that of well-known baseline algorithms from
the literature. Since we are the first to tackle the problem of
\mcproblem, the baseline algorithms we compare with do not target our objective
function directly. Nevertheless, the comparison helps us highlight the settings
in which our algorithms are essential to achieve good performance
for \mcproblem.

Below, we describe the respective baselines for the two variants of the problem.

\emph{Baselines for {\nodeproblem}:} For a budget $k$, the following baselines
return a set of $k$ nodes with highest value for the respective measure:
%\begin{itemize}
\squishlist
	\item {\indegree}: number of incoming edges;
	\item {\inprobability}: total probability of incoming edges;
	\item {\nodebetweenness}: as defined in~\cite{brandes01faster,erdos15divide,riondato16fast};
	\item {\closeness}: as defined in~\cite{sabidussi1966centrality};
	\item {\nodenumitems}: number of items before transition;
	%\item {\pagerank}: PageRank of the node~\cite{pagerank1999}.
%\end{itemize}
\squishend
	% The {\betweenness} and {\closeness}
	% heuristics choose top-$k$ nodes sorted using their betweenness and 
	% closeness centrality measures
	% respectively. Finally, the {\numitems} baseline chooses $k$ nodes with 
	% highest items.

\emph{Baselines for {\edgeproblem}:}. For a budget $k$, the following baselines
return a set of $k$ edges with highest value for the respective measure:
%\begin{itemize}
\squishlist
	\item {\edgebetweenness}: as defined in~\cite{brandes01faster,erdos15divide,riondato16fast};
	\item {\edgenumitems}: expected number of items to transition over the edge;
	\item {\probability}: transition probability of the edge.
\squishend
%\end{itemize}

% \spara{Note:} The same name is used for some baselines across problem variants
% (\betweenness, \numitems).
In what follows, context determines which baseline and variant we refer to.

\begin{table*}[htbp]\small
\centering
\begin{tabular}{@{}llcccc@{}}
\toprule
\multirow{2}{*}{\begin{tabular}[l]{@{}l@{}}Graph\end{tabular}}                                                                         & \multirow{2}{*}{\begin{tabular}[l]{@{}l@{}}Item\\Distribution\ \ \ \ \ \ \end{tabular}} & \multirow{2}{*}{\begin{tabular}[l]{@{}l@{}}$r({\nodegreedy})$\end{tabular}} & \multirow{2}{*}{\begin{tabular}[l]{@{}l@{}}$r({\tt Node-Baseline}^\ast)$\end{tabular}} & \multirow{2}{*}{\begin{tabular}[l]{@{}l@{}}$r({\edgegreedy})$\end{tabular}} & \multirow{2}{*}{\begin{tabular}[l]{@{}l@{}}$r({\tt Edge-Baseline}^\ast)$\end{tabular}}\\ \\ \toprule
\multirow{4}{*}{\begin{tabular}[l]{@{}l@{}}{\autonomoussystems}\end{tabular}}  & {\ego}               & 0.06                                  & 0.24                                   & 0.14                                  & 0.15                                   \\
                                                                               & {\direct}            & 0.66                                  & 0.67                                   & 0.99                                  & 0.99                                   \\
                                                                               & {\uniform}           & 0.38                                  & 0.40                                   & 0.97                                  & 0.99                                   \\
                                                                               & {\inverse}           & 0.38                                  & 0.40                                   & 0.97                                  & 0.99                                   \\ \cmidrule(l){2-6}
\multirow{4}{*}{\begin{tabular}[l]{@{}l@{}}{\geo}\end{tabular}} 			   & {\ego}               & 0.00                                  & 0.06                                   & 0.01                                  & 0.02                                   \\
                                                                               & {\direct}            & 0.00                                  & 0.06                                   & 0.20                                  & 0.65                                   \\
                                                                               & {\uniform}           & 0.00                                  & 0.06                                   & 0.15                                  & 0.65                                   \\
                                                                               & {\inverse}           & 0.00                                  & 0.07                                   & 0.15                                  & 0.65                                   \\ \cmidrule(l){2-6}
\multirow{4}{*}{\begin{tabular}[l]{@{}l@{}}{\grid}\end{tabular}} 	           & {\ego}               & 0.27                                  & 0.27                                   & 0.29                                  & 0.29                                   \\
                                                                               & {\direct}            & 0.92                                  & 0.92                                   & 0.98                                  & 0.98                                   \\
                                                                               & {\uniform}           & 0.92                                  & 0.92                                   & 0.98                                  & 0.98                                   \\
                                                                               & {\inverse}           & 0.92                                  & 0.92                                   & 0.98                                  & 0.98                                   \\ \cmidrule(l){2-6}
\multirow{4}{*}{\begin{tabular}[l]{@{}l@{}}{\ba}\end{tabular}} 	               & {\ego}               & 0.18                                  & 0.56                                   & 0.26                                  & 0.26                                   \\
                                                                               & {\direct}            & 0.71                                  & 0.71                                   & 0.99                                  & 0.99                                   \\
                                                                               & {\uniform}           & 0.63                                  & 0.63                                   & 0.98                                  & 0.98                                   \\
                                                                               & {\inverse}           & 0.63                                  & 0.63                                   & 0.98                                  & 0.98                                   \\ \bottomrule                                                                             
\end{tabular}
\caption{Comparison of greedy algorithms with the best-performing baseline (${\tt Node-Baseline}^\ast$ and ${\tt Edge-Baseline}^\ast$) for $k=50$. 
For a given pair of graph and item-distribution scheme, $r(A)$ expresses the 
ratio of the expected uncertainty that algorithm $A$ achieves with $k=50$
monitoring operations over the initial uncertainty $\uncertainty_0$ 
(for $k=0$). Note that the best-performing baseline is different for different rows of the table.
}
\label{tab:objective-evolution-table}
\end{table*}

\begin{figure}
\begin{subfigure}[b]{0.23\textwidth}
\includegraphics[width=\textwidth]{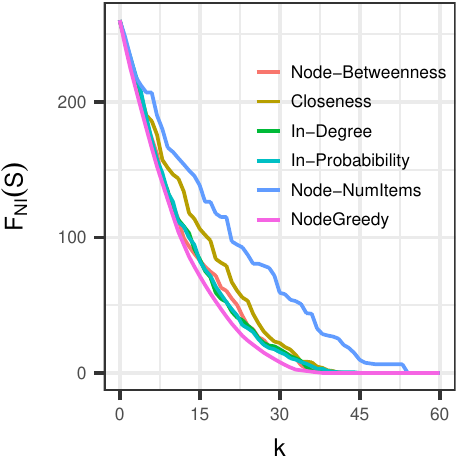}
\caption{{\nodeproblem}}
\label{fig:hubway_nodes}
\end{subfigure}
\begin{subfigure}[b]{0.23\textwidth}
\includegraphics[width=\textwidth]{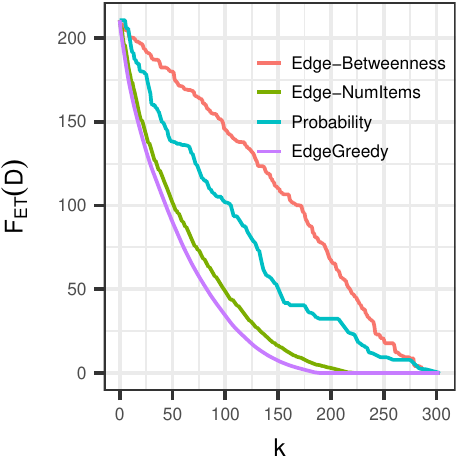}
\caption{{\edgeproblem}}
\label{fig:hubway_edges}
\end{subfigure}
\caption{{\hubway} data; $y$-axis: expected uncertainty, $x$-axis: number of monitored nodes or edges.}
\end{figure}

\begin{figure*}
\centering
\includegraphics{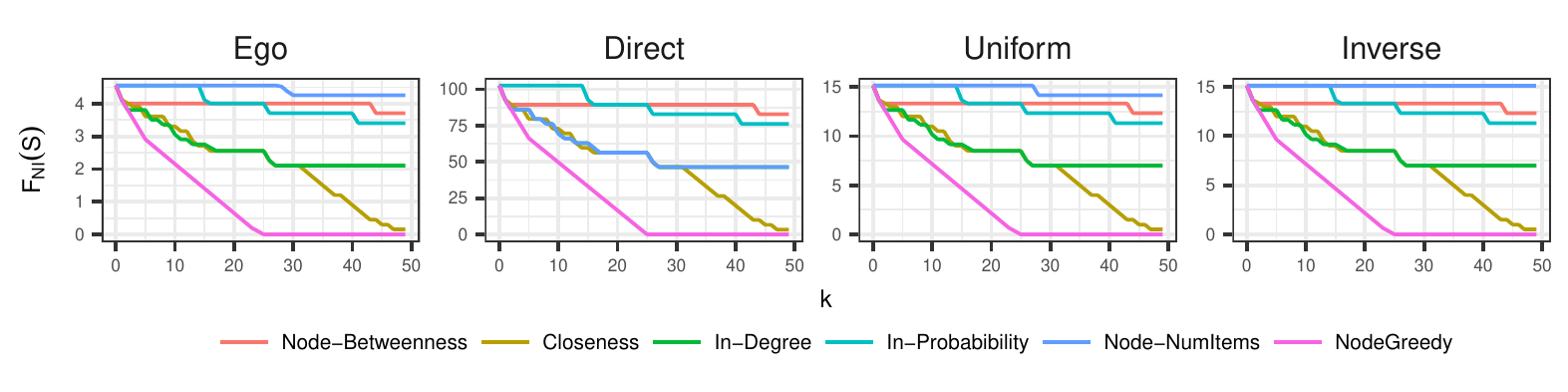}
\caption{{\nodeproblem} {\geo} dataset; $y$-axis expected uncertainty, 
$x$-axis: number of monitored nodes ($k$).}
\label{fig:geo_nodes}
\end{figure*}

\begin{figure*}
\centering
\includegraphics{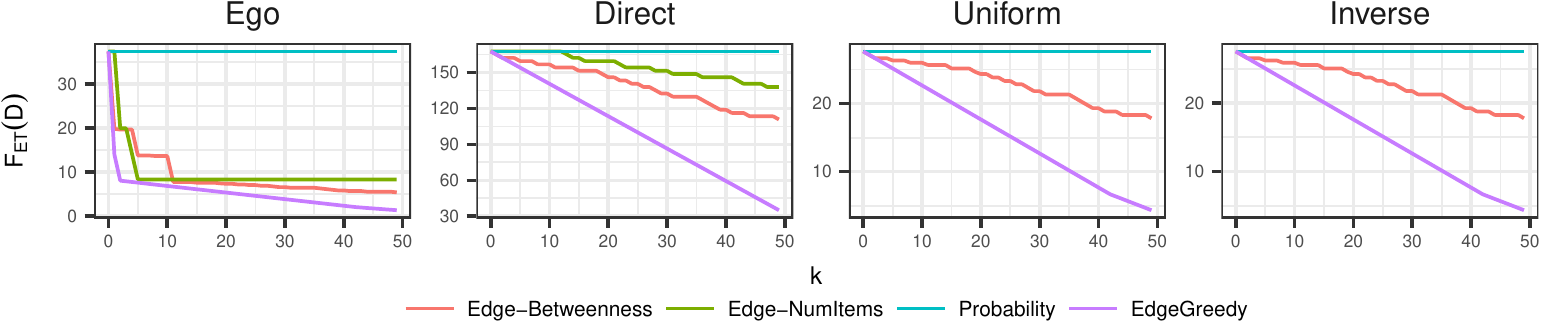}
\caption{{\edgeproblem} {\geo} dataset; $y$-axis expected uncertainty, 
$x$-axis: number of monitored edges ($k$).}
\label{fig:geo_edges}
\end{figure*}

\begin{figure*}
\centering
\includegraphics{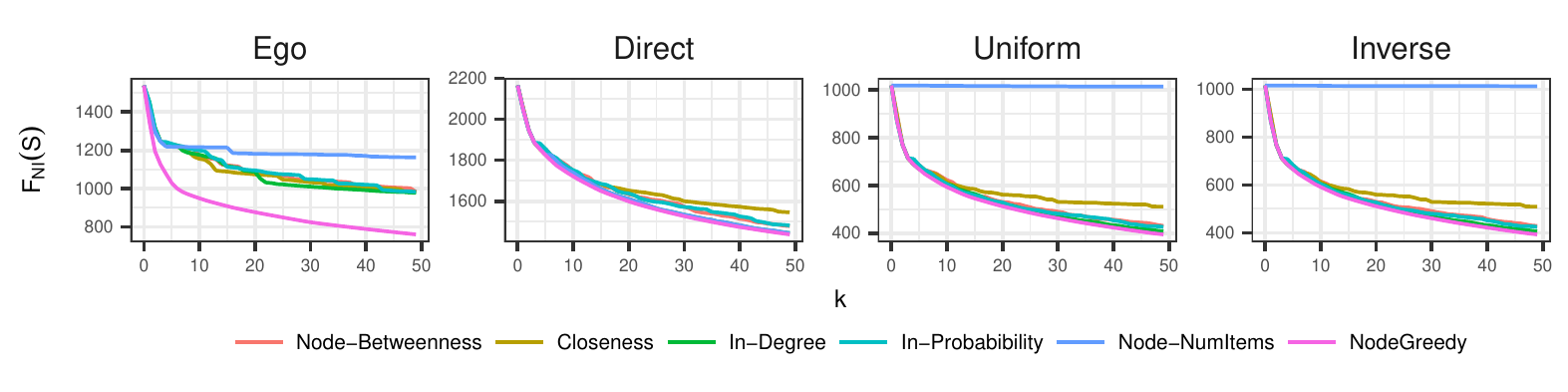}
\caption{{\nodeproblem} {\autonomoussystems} dataset; $y$-axis expected uncertainty, 
$x$-axis: number of monitored nodes ($k$).}
\label{fig:ass_nodes}
\end{figure*}

\begin{figure*}
\centering
\includegraphics{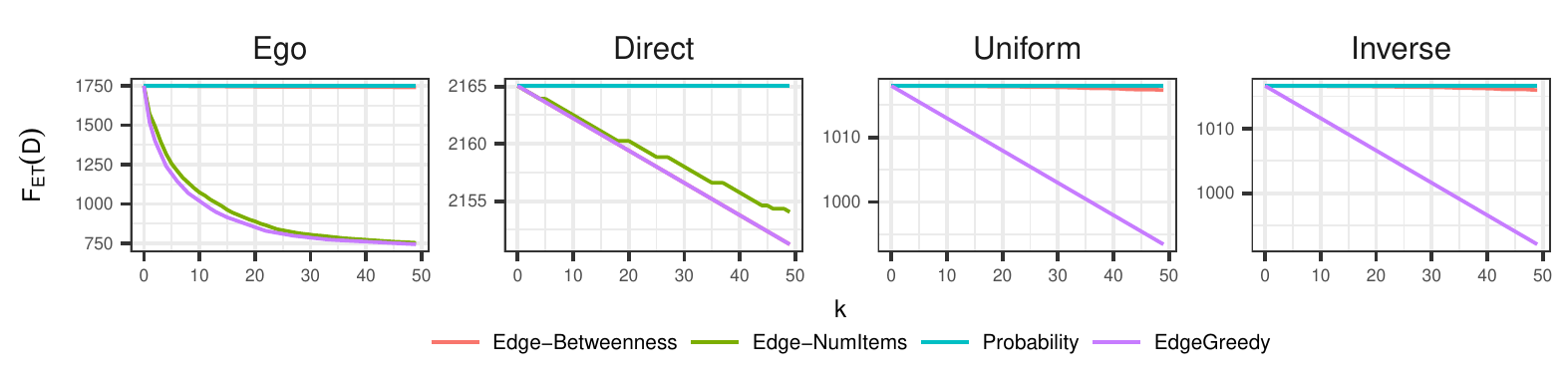}
\caption{{\edgeproblem} {\autonomoussystems} dataset; $y$-axis expected uncertainty, 
$x$-axis: number of monitored edges ($k$).}
\label{fig:ass_edges}
\end{figure*}

\begin{figure*}
\centering
\includegraphics{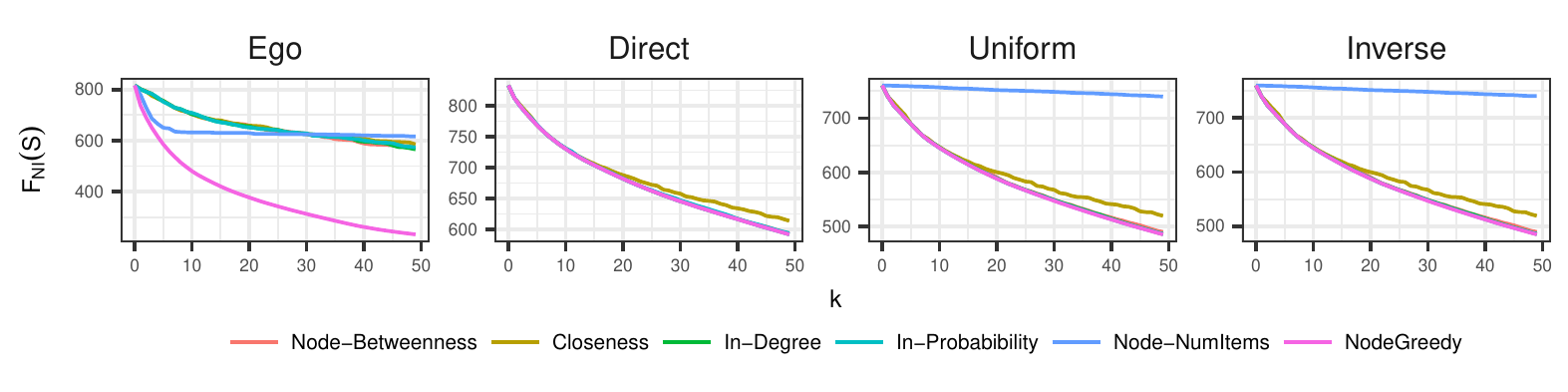}
\caption{{\nodeproblem} {\ba} dataset; $y$-axis expected uncertainty, 
$x$-axis: number of monitored nodes ($k$).}
\label{fig:ba3_nodes}
\end{figure*}

\begin{figure*}
\centering
\includegraphics{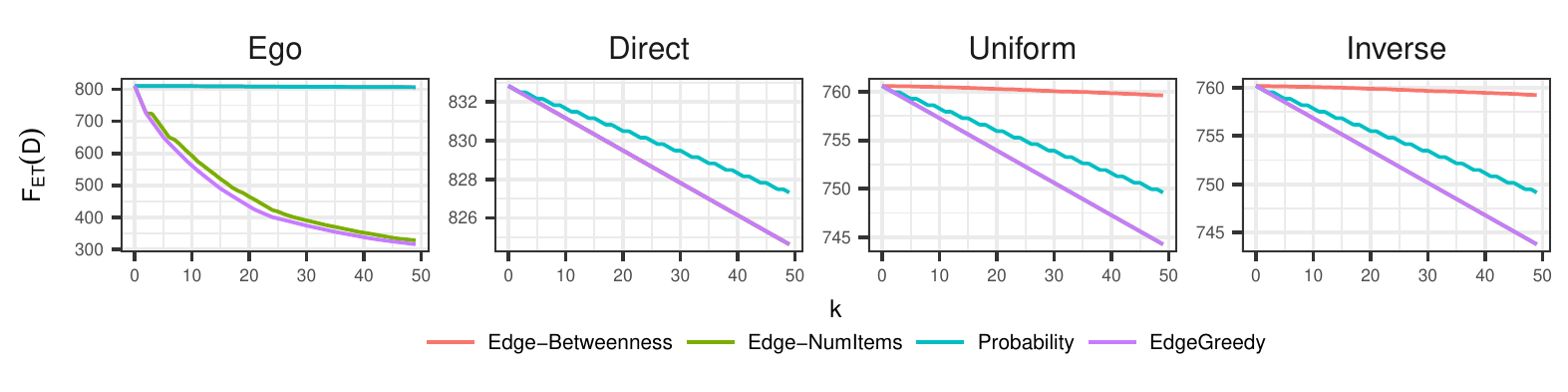}
\caption{{\edgeproblem} {\ba} dataset; $y$-axis expected uncertainty, 
$x$-axis: number of monitored edges ($k$).}
\label{fig:ba3_edges}
\end{figure*}

\begin{figure*}
\centering
\includegraphics{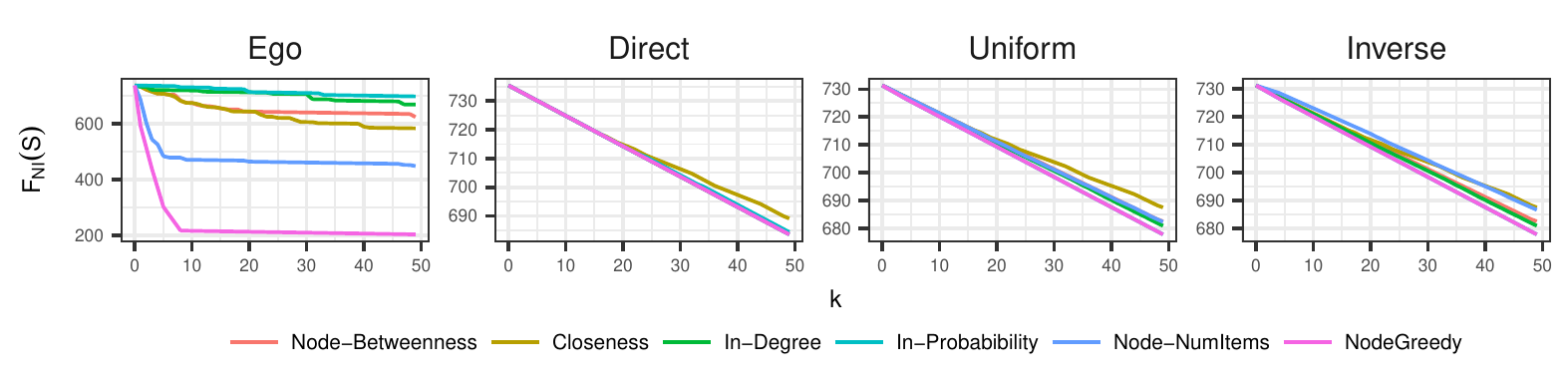}
\caption{{\nodeproblem} {\grid} dataset; $y$-axis expected uncertainty, 
$x$-axis: number of monitored nodes ($k$).}
\label{fig:grid_nodes}
\end{figure*}

\begin{figure*}
\centering
\includegraphics{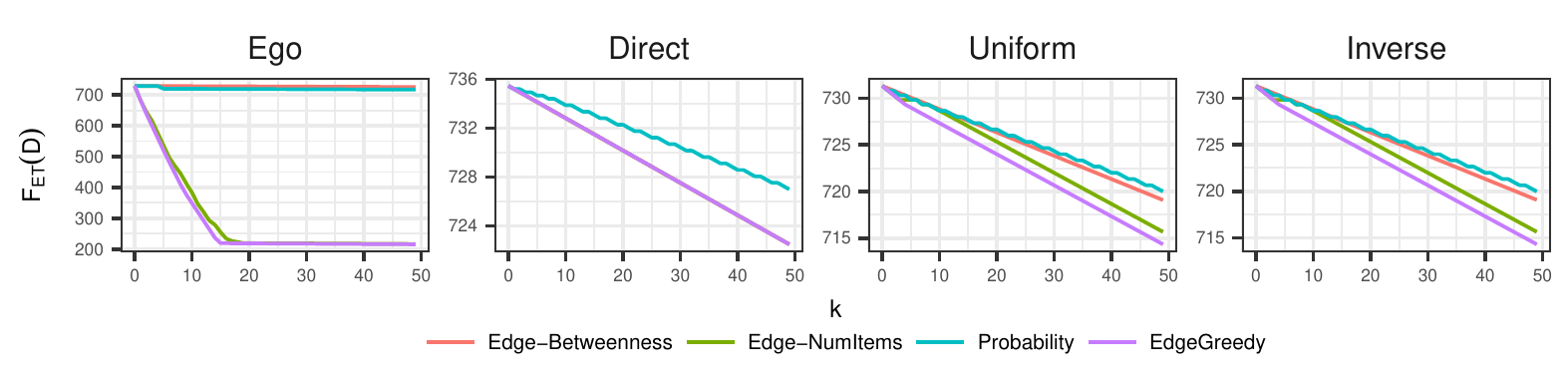}
\caption{{\edgeproblem} {\grid} dataset; $y$-axis expected uncertainty, 
$x$-axis: number of monitored edges ($k$).}
\label{fig:grid_edges}
\end{figure*}

\subsection{Experimental results}
In this section, we report the   
performance of algorithms for the {\mcproblem} problem -
first on the graph datasets, combined with item distribution schemes;
then on the \hubway\ dataset.
As objective we always use the expected uncertainty
achieved for a given budget $k$ of nodes or edges -- the smaller its value,
the better the performance of the algorithm.
Note that we do not report separately the performance of \edgeDP, as it achieves
same performance as \edgegreedy, but is not as efficient.

We provide the results for the graph datasets in Table \ref{tab:objective-evolution-table}. 
In all these experiments we use $k=50$.
Moreover, $r(A)$ is the ratio of the achieved objective value (for $k=50$) over
the initial value $\uncertainty_0$ of the measure 
(for no monitoring operations, i.e., $k=0$).
The table shows four quantities for every graph-item distribution pair
: $r(A)$ for $A=\{${\nodegreedy}, {\tt Node-Baseline}$^\ast$, {\edgegreedy}, {\tt Edge-Baseline}$^\ast$ $\}$.
Note that {\tt Node-Baseline}$^\ast$ (resp.\ {\tt Edge-Baseline}$^\ast$) refers to the baseline 
algorithm with the best performance.
For every algorithm $A$, $r(A)\in[0,1]$ and the smaller the value of $r(A)$ the better the performance of the algorithm.
 
From the table, we observe that for 
the {\autonomoussystems} dataset,  {\nodegreedy}  significantly outperforms the best
baseline for the {\ego} item distribution, while performing marginally better for other item distributions.
The value of  $r({\edgegreedy})$ is only slightly less than
the best baselines across all the configurations. However, we observe that there is no baseline
which performs uniformly the best across different item distributions. For example,  {\edgebetweenness}
is the best baseline for {\direct} item distribution, the {\edgenumitems} for {\ego}, while they both
perform worse than even randomly chosen edges for {\uniform} and {\inverse} item distributions.
Notably, for the {\geo} graphs, the greedy algorithms significantly outperform the baselines

For the {\grid} graphs, the baselines perform exactly
the same as our algorithms. This can be explained by the nature of the 
{\grid} graph, where all the nodes except the ones on the boundary are similar to each other,
thereby rendering the {\direct}, {\uniform} and {\inverse} item distributions very similar 
to each other. For the {\ego} distribution, the greedy algorithms perform marginally better
than the baselines.
%, depending on the number of nodes within the fixed radius around the randomly chosen node, 
%and the percentage of items distributed among these nodes. 
Again,
there is no baseline which performs uniformly the best.
Similar is our explanation for the results on {\ba} graphs as in these graphs most of the nodes have
almost the same (small) degrees too. 
% \footnote{A more thorough investigation of the performance of the algorithms for different
% values of $k$ and for all datasets we consider is shown in 
% Supplementary Material, Section~\ref{sec:additional_results}.
% }
% \subsection{Additional Results}
% \label{sec:additional_results}
%\subsection{Performance of greedy algorithms on {\geo} graphs.}
Figure~\ref{fig:geo_nodes}
shows the performance of the {\nodegreedy} algorithm for the 
the {\geo} graphs, with each plot corresponding to a different item
distribution schemes.
Observe that {\nodegreedy} significantly outperforms
all other baselines, which capture different semantics of centrality.
In particular, we observe that {\nodegreedy} achieves zero or near-zero
expected uncertainty with a small fraction of selected nodes compared to baselines.
Among the baselines, {\closeness} performs second-best in many cases,
while {\indegree} performs as well as {\closeness} for small $k$.

Similarly, Figure~\ref{fig:geo_edges} shows the 
performance of the different algorithms for the {\edgeproblem} and the {\geo} 
graphs, for all
possible item-distribution schemes.
As before, we observe that {\edgegreedy} outperforms the baselines in all 
cases. We notice also that the pattern of performance differs somewhat for the 
case of {\ego} item distribution. With the exception of one 
baseline ({\probability}), all algorithms achieve steep decline in expected uncertainty for
small value of $k$ - {\edgegreedy} performs best, but baselines are competitive. 
However, for larger $k$, the performance of baselines  
does not keep up with that of {\edgegreedy}.
We believe that this is can be explained as follows:
the first edges selected by 
baselines are either central in terms of graph structure 
-- and therefore near the part of the graph with high concentration of items 
({\edgebetweenness}) -- or directly in the area of the graph with many items 
({\edgenumitems}). 
In terms of reducing expected uncertainty, this is beneficial at first. However, these baselines
as they do not optimize our objective are not able to continue reducing the expected
uncertainty with their subsequent selections.

Figure~\ref{fig:ass_nodes} and Figure~\ref{fig:ass_edges} show the performance of 
the greedy algorithms on the {\nodeproblem} and the {\edgeproblem} problems
respectively. We observe that both the {\nodegreedy} and the {\edgegreedy} algorithms
are consistently the best when compared to the baselines. However, $k=50$ represents
about 1\% of the total edges in the graph, hence their monitoring does not decrease
the uncertainty significantly.
While experiments with larger values of $k$ are prohibitive due to time complexity of the {\edgegreedy} algorithm, we postulate that the greedy algorithm will still continue
outperforming the baselines.

Figures~\ref{fig:ba3_nodes} and ~\ref{fig:ba3_edges} provide a similar comparison
for the different configurations of the {\ba} graph. The greedy algorithms
provide marginal benefits or perform on par with competitive baselines. On the
{\ba} graphs, for {\direct}, {\uniform} and {\inverse} item distributions, some baselines
perform exactly the same as the greedy algorithms for relatively small number of
monitoring operations i.e., $k=50$. Lastly, we observe similar trends in case of the
{\grid} graphs as evident in Figures~\ref{fig:grid_nodes} and ~\ref{fig:grid_edges}.
It should be noted that there is no baseline method that provides a consistently competitive
performance with the greedy algorithms across all different configurations described above.
\balance

\spara{Experiments with {\hubway} data:}
In our last experiment, 
we explore the performance of our algorithms on the {\hubway} dataset. From Figure \ref{fig:hubway_nodes}
and Figure \ref{fig:hubway_edges} we observe that the {\nodegreedy} and {\edgegreedy} algorithms
are consistently the best at reducing expected uncertainty, although the baselines are competitive on the
relatively smaller graph. In Figure \ref{fig:hubway_stations}, we plot the Hubway stations across Boston
chosen by the {\nodegreedy} algorithm with $k=5$. The nodes chosen by the algorithm are supported by
the anecdotal evidence of being exactly some of the of the most popular landmarks around the city.
From a managerial perspective, tracking the number of trips starting or ending at these Hubway stations
can help the operators better reduce the expected uncertainty around the expected number of bikes
available at its different stations and anticipate future bike ``re-balancing"\footnote{https://www.citylab.com/transportation/2014/08/balancing-bike-share-stations-has-become-a-serious-scientific-endeavor/379188/} operations.

\spara{Running times:} For all our experiments we use a single process 
implementation of our algorithms on a 24-core 2.9GHz Intel Xeon E5 processor 
with 512GB memory. For the  largest graph in our experiments,  
the parallelized version of the {\nodegreedy}\ takes about $5-10$ 
seconds per selected node, while the parallelized version of {\edgegreedy}\ 
takes about $1$ minute per selected edge.
\begin{figure}
	\centering
	\includegraphics[width=8cm, height=4cm]{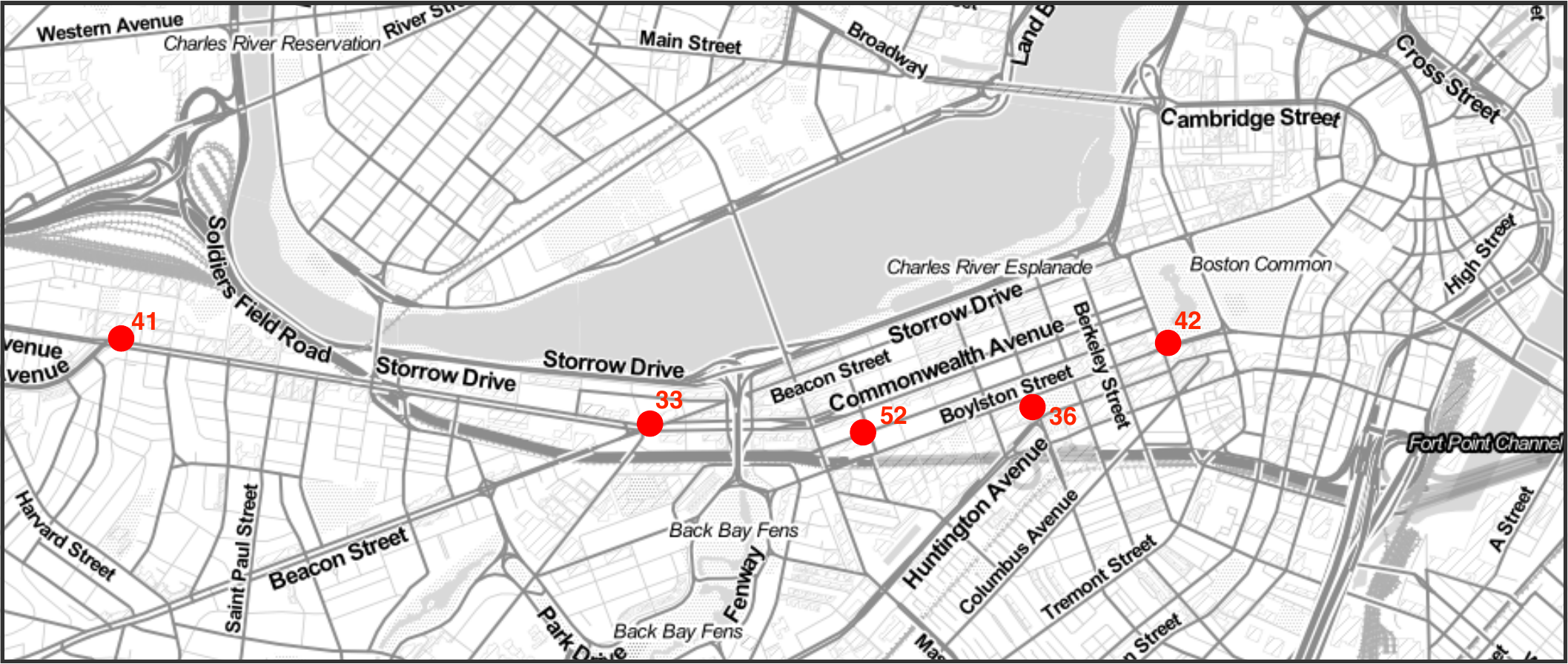}
	\caption{{\hubway} data. IDs of stations picked as a solution to {\nodeproblem}
	for $k=5$;
	33: Kenmore Sq., 36: Copley Sq./Boston Public Library, 41: Packard's Corner, 
	42: Boston Public Garden, 52: Newbury St.                        }
	\label{fig:hubway_stations}
\end{figure}

\spara{Discussion:} Our experiments show that  {\nodegreedy} and  {\edgegreedy} 
consistently perform better than or on par with other popular baseline methods. Also, for 
graphs with relatively large number of nodes, the solutions to the {\nodeproblem} problem 
are more effective at reducing
the expected uncertainty than the solutions to the {\edgeproblem} problem for the same number of
node (resp.\ edge) monitors.
This is
especially important considering our analysis from Section \ref{sec:nodes} and \ref{sec:edges} which show
that the {\nodegreedy} algorithm has a better time complexity compared to the {\edgegreedy} for dense graphs.

%!TEX root = arxiv.tex
\section{Conclusions}

In this paper, we introduced the problem of \mcproblem: given
a distribution of items over a \markovchain, we aim to perform
a limited number of monitoring operations, so as to adequately
predict the position of items on the chain after one transition step.
We studied variants of this problem
and provided efficient algorithms to solve them. Our experimental evaluation
demonstrated the superiority of the proposed algorithms compared to 
baselines and the practical utility of the results in real settings.
%We see this paper as a first step of a line of works to come.
%Notably, we addressed the problem in a setting that assumed we had full 
%knowledge of item distribution before transition. 
A natural extension of this work is to to select monitoring operations under incomplete information
for initial item distribution
-- which would allow the operations to be deployed in perpetuity.
Another future-work direction is to consider different types of 
monitoring operations such as those
that combine knowledge of item placement on nodes
and edges. Finally, we can also consider more complex traffic models
(e.g., involving queuing and different transition delays between 
nodes~\cite{gallager2012discrete}).
In all cases, the algorithms we developed in this paper will serve as 
the basis for future work.

\todo[MM]{An interesting meta-question is to see which algorithm
leads to better objective function: greedy for \parentstransitions\
or for \childrentransitions?}
\bibliographystyle{plain}
\bibliography{bibliography}
\clearpage

\end{document}